\documentclass[twocolumn]{autart}    

\usepackage{graphicx}          

\usepackage{cite}
\usepackage{wrapfig}
\usepackage{multirow}

\usepackage{psfrag}
\usepackage{latexsym, amsmath, color, amsfonts, amssymb}
\usepackage{algorithm}
\let\classAND\AND
\let\AND\relax
\usepackage{algorithmic}

\let\AND\classAND
\AtBeginEnvironment{algorithmic}{\let\AND\algoAND}

\newtheorem{lemma}{Lemma}

\newtheorem{proposition}{Proposition}
\newtheorem{remark}{Remark}
\newtheorem{assumption}{Assumption}
\newtheorem{problem}{Problem}
\newtheorem{proof}{Proof}

\begin{document}

\begin{frontmatter}

\title{Sensor Selection for Remote State Estimation with QoS Requirement Constraints} 

\thanks[footnoteinfo]{This paper was not presented at any IFAC meeting. Corresponding author: Lingying~Huang.}
\thanks[footnoteinfo]{The work by H. Yang and L. Shi is supported by a Hong Kong RGC General Research Fund 16206620. The work by Y. Mo is supported by National Natural Science Foundation of China under grant no. 62273196.}

\author[hkust]{Huiwen Yang}\ead{hyangbr@connect.ust.hk}, 
\author[ntu]{Lingying Huang}\ead{lingying.huang@ntu.edu.sg},
\author[ecust]{Chao Yang}\ead{yangchao@ecust.edu.cn},  
\author[thu]{Yilin Mo}\ead{ylmo@tsinghua.edu.cn},
\author[hkust]{Ling Shi}\ead{eesling@ust.hk}

\address[hkust]{Department of Electronic and Computer Engineering, Hong Kong University of Science and Technology, Hong Kong, China}  
\address[ntu]{School of Electrical and Electronic Engineering, Nanyang Technological University, Singapore}             
\address[ecust]{Key Laboratory of Smart Manufacturing in Energy Chemical Process, Ministry of Education, East China University of Science and Technology, Shanghai, China}        
\address[thu]{Department of Automation and BNRist, Tsinghua University, Beijing, China}


\begin{abstract}                          
	In this paper, we study the sensor selection problem for remote state estimation under the Quality-of-Service (QoS) requirement constraints. Multiple sensors are employed to observe a linear time-invariant system, and their measurements should be transmitted to a remote estimator for state estimation. However, due to the limited communication resources and the QoS requirement constraints, only some of the sensors can be allowed to transmit their measurements. To estimate the system state as accurately as possible, it is essential to select sensors for transmission appropriately.
    We formulate the sensor selection problem as a non-convex optimization problem. It is difficult to solve such a problem and even to find a feasible solution. To obtain a solution which can achieve good estimation performance, we first reformulate and relax the formulated problem. Then, we propose an algorithm based on successive convex approximation (SCA) to solve the relaxed problem. By utilizing the solution of the relaxed problem, we propose a heuristic sensor selection algorithm which can provide a good suboptimal solution. Simulation results are presented to show the effectiveness of the proposed heuristic.
\end{abstract}

\end{frontmatter}

\section{Introduction}
	Recently, many promising real-life applications such as smart city, agriculture, and transportation, have benefited from the development of cyber-physical systems (CPS) and Internet-of-Things (IoT). 
	Remote state estimation plays a very important role in wireless CPS, where sensors can transmit their measurements to a remote estimator via wireless channels and the remote estimator utilizes the received measurement data to estimate the states of a system monitored by the sensors. 
	As the manufacturing cost of sensors has been reduced, an increasing number of sensors have been deployed to expand surveillance coverage and thus provide access to richer data. 
	However, communication resources (spectrum, transmission power, etc.) are usually limited. As a result, it is impossible for a wireless communication system to permit the transmission of all the sensors at the same time, especially when the quantity of sensors is massive. 
	
	Existing works in the control community have extensively investigated the sensor scheduling and resource allocation for remote state estimation with limited communication resources~\cite{wang2019whittle, wu2020optimal}. 
	Wang et~al.~\cite{wang2019whittle} studied the multichannel allocation problem when the number of available channels is less than the number of sensors. 
	Wu et~al.~\cite{wu2020optimal} studied the optimal sensor scheduling policy when the communication channels have limited bandwidth. 
	These works formulated their sensor scheduling problems as a Markov decision process (MDP), which is computationally inefficient in multisensor scenarios due to the curse of dimensionality. To handle this problem, they proposed an index-based heuristic to provide asymptotically optimal policy instead of solving the MDP using some numerical algorithms such as value iteration and policy iteration. 
	Some other sensor scheduling problems are formulated as deterministic sensor selection problems~\cite{mo2011sensor, shi2013optimal, shi2013approximate, yang2015deterministic, asghar2017complete, huang2021joint}, which are usually non-convex due to the inevitable introduction of integer constraints. As a result, they either relaxed the non-convex problems to convex programming problems~\cite{mo2011sensor}, or proposed some periodic schedule heuristic, e.g., a branch-and-bound-based
	algorithm~\cite{shi2013optimal}, a dynamic programming based approach~\cite{shi2013approximate}, a greedy-based algorithm~\cite{asghar2017complete}, etc. 
	
    The works mentioned above all considered the scenarios where sensors only have two choices: whether transmit their measurements or not. However, in real communication systems, sensors can realize more flexible transmission by adjusting their transmission power. Meanwhile, the multiple access technique makes it possible to share a limited amount of radio spectrum among multiple users~\cite{tse2005fundamentals}, which facilitates the full utilization of communication resources and the promotion of transmission efficiency. One of the fundamental issues for the implementation of multiple access is interference management~\cite{hossain2014evolution}. Since users are allowed to access the same resource block at the same time, the transmission power of a sensor can cause interference with the transmission of other sensors.  
    {Li~et~al.~\cite{li2014multi} established a simple game framework for multi-sensor remote state estimation with an energy constraint. Li~et~al.~\cite{li2019power} formulated the power control problem for multi-sensor remote state estimation as a stochastic programming problem, where only non-negative transmission power constraints were considered. 
    Ding~et~al.~\cite{ding2021interference} investigated the interference management for a CPS with primary sensors and potential sensors, whose interaction was formulated as a non-cooperative game. The works mentioned above adopted the signal-to-interference-and-noise ratio (SINR) model to characterize the interference among sensors, and they all considered that the packet dropout rate is a non-increasing function in SINR. The packet dropout rate was used to calculate the expected estimation error covariance, which appeared in the objectives of the studied problems. 
    However, by utilizing modulation and coding, error-free decoding can be achieved if and only if the SINR at the receiver side is above a threshold, which is determined by the adopted modulation and coding schemes~\cite{dobre2011second}. Such a predefined threshold is called Quality-of-Service (QoS) requirement~\cite{li2007real}. 
    In this scenario, the reception of the signals transmitted by sensors becomes a deterministic process rather than a random process, which is significantly different from the previous problem formulations. 
   }
    
	In the communication community, researchers have made great efforts to investigate interference management problems with QoS requirement constraints~\cite{li2007real, zhao2017joint, tang2019user, yang2021joint, xia2021joint}. 
	As the mutual interferences among users degrade the SINR performance, a communication system probably is unable to ensure that all users can transmit their data successfully, i.e., not all users' QoS requirements can be satisfied, which is reflected in the fact that there may be no feasible solution for the optimization problem considering the satisfaction of all uses' QoS requirements. Therefore, it is necessary to select a subset of users, whose QoS requirements can be satisfied simultaneously, to transmit their data. 
	Tang and Feng~\cite{tang2019user} proposed a minimum-mean-square-error (MMSE) based user selection algorithm whose core idea is to minimize the gaps between users' achievable SINR and their QoS requirements. Later, Xia et~al.~\cite{xia2021joint} jointly selected users and designed the transmitter parameters based on a similar idea as~\cite{tang2019user}.
	However, they can only maximize the number of users whose QoS requirements can be satisfied, which is not necessary for remote state estimation with multiple sensors. The reason is that the number of selected sensors is not the determining factor of the estimation performance. Intuitively, it can be more worthy to select one sensor that has accurate measurements rather than select several sensors whose measurements are extremely inaccurate. Therefore, the existing user selection algorithms in the communication community are inapplicable for the sensor selection for remote state estimation. 
	
	In this paper, we consider the remote state estimation with real-time QoS requirement constraints. The objective of the remote estimator is to estimate the system state as accurately as possible by collecting measurement data from different sensors. However, due to the QoS constraints, the communication system may be unable to guarantee the transmission of all the sensors at the same time. Therefore, at each time slot, the sensors whose transmission is allowed should be selected according to the current estimation error covariance, channel states, and real-time QoS requirements. The main contributions of this paper are summarized as follows:
	\begin{itemize}
		\item[1)] We study the sensor selection problem for remote state estimation under the constraints of sensors' QoS requirements. This problem exists in some practical communication systems, but has not yet been considered in the existing works. 
        Due to the QoS requirement constraints, both continuous variables (transmission power) and discrete variables (decision of selection) should be optimized in the sensor selection problem. Moreover, the QoS requirements introduce more nonconvexity to the problem. Compared with the existing sensor selection problems where only discrete (binary) decision variables were optimized, our problem is more complicated and computationally expensive. {Compared with~\cite{tang2019user, yang2021joint, xia2021joint}, which only maximize the number of users, we take the accuracy of sensors' measurements into account. As a result, a matrix optimization variable, which has a coupling relationship with the discrete decision variables, is introduced to characterize the estimation error. Hence, it becomes challenging to solve the sensor selection problem.}
		\item[2)] We formulate the sensor selection problem as a non-convex optimization problem with equality of matrices and integer (binary) constraints, and provide a heuristic to solve the formulated problem. Specially, we first adopt linear relaxation to relax the binary constraints and transform the equality of matrices into a linear matrix inequality. Then, we adopt the successive convex approximation to transform the relaxed problem into a series of convex optimization problems. By utilizing the solution of the relaxed problem obtained by the algorithm based on SCA, we propose a heuristic sensor selection algorithm based on the concept of assimilated sensing precision matrix. Simulation results show that the proposed heuristic outperforms the existing methods.
	\end{itemize}


	{
		The remainder of this paper is organized as follows.
		Section \uppercase\expandafter{\romannumeral2} presents the mathematical setup and description of the problem. In Section \uppercase\expandafter{\romannumeral3}, an algorithm for solving the formulated problem is provided and analyzed. In Section \uppercase\expandafter{\romannumeral4}, the simulation results are presented to verify the effectiveness of the algorithm proposed in Section \uppercase\expandafter{\romannumeral3}. Section \uppercase\expandafter{\romannumeral5} concludes this paper and presents some future work.}
	
	{
		\emph{Notations:} 
		$\mathbb{R}$ is the set of real numbers, $\mathbb{R}^n$ is the $n$-dimensional Euclidean space, and $\mathbb{R}^{n\times m}$ is the set of real matrices with size $n\times m$. 
		For a matrix $X$, $X > 0$ $(X\geq 0)$ denotes $X$ is a positive definite (positive semidefinite) matrix. $\mathrm{Tr}\{\cdot\}$ denotes the trace of a matrix. 
		$\mathbb{E}[\cdot]$ is the expectation of a random variable. $\mathbb{E}[\cdot | \cdot]$ refers to conditional expectation.
	}
	
	\begin{figure}[]
		\centering
		\includegraphics[width=3.4in]{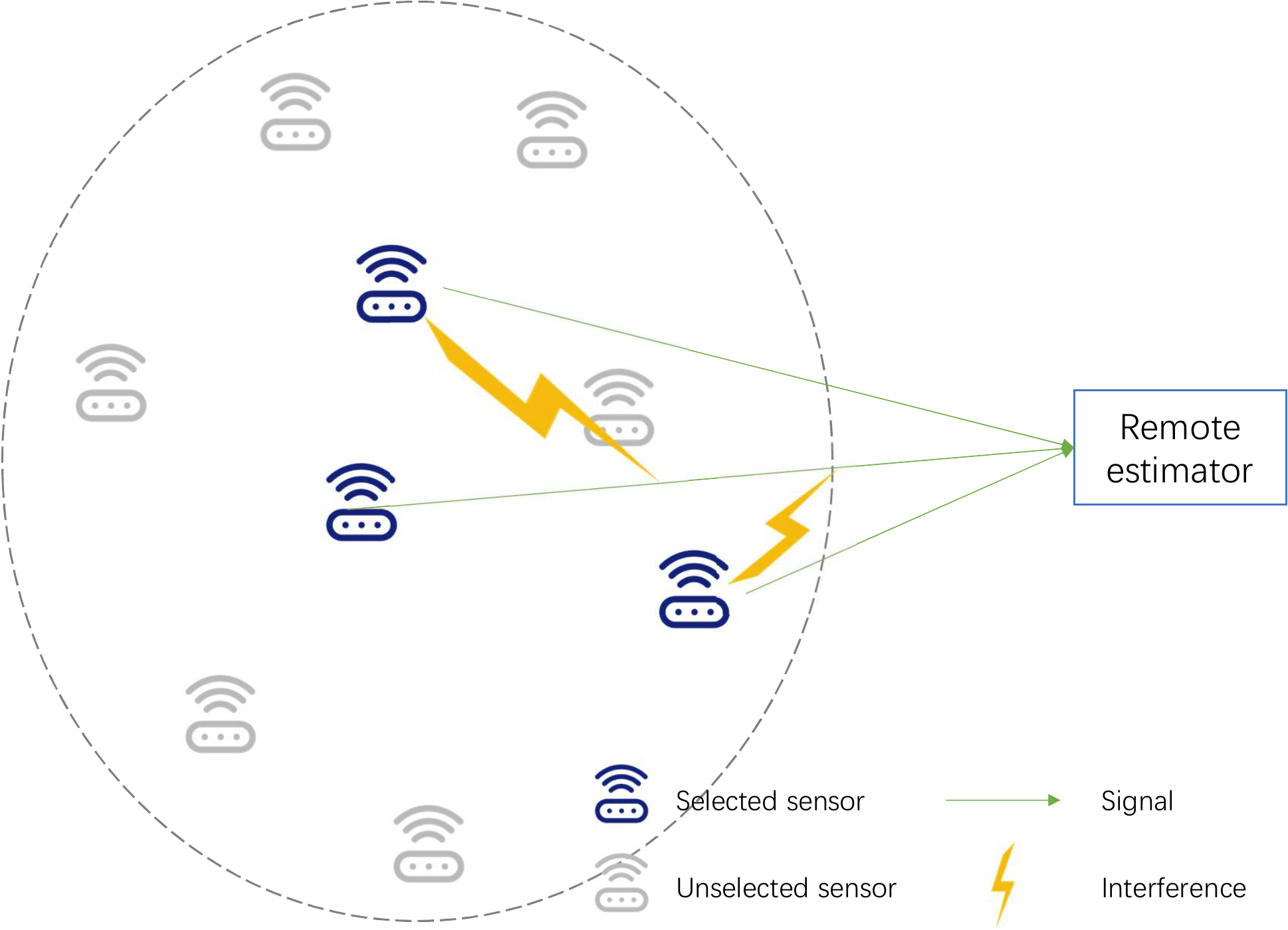}
		\caption{An instance of sensors' transmission}
		\label{system}
	\end{figure}
	\section{Problem Setup}
	\subsection{System Model}
	\begin{figure}[]
		\centering
		\includegraphics[width=3.4in]{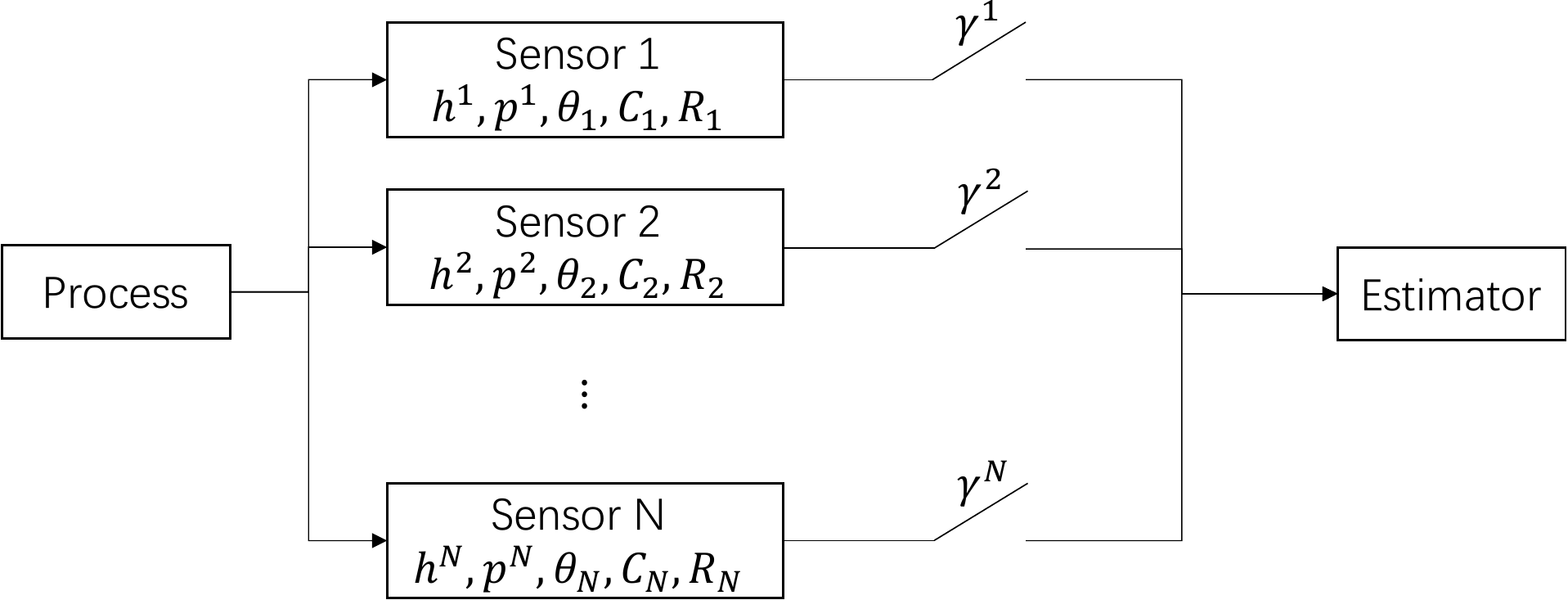}
		\caption{{System model}}
		\label{diagram}
	\end{figure}
	We consider a discrete-time linear time-invariant (LTI) system monitored by $N$ sensors:
	\begin{align}
	x_{k+1} &= A x_k + w_k,\\
	y_k^i &= C_i x_k + v_k^i, \quad i = 1, 2, \ldots, N,
	\end{align}
	where $x_k \in \mathbb{R}^n$ is the system state at time $k$ and $y_k^i\in\mathbb{R}^{m_i}$ is the measurement taken by sensor $i$.
	$w_k \in \mathbb{R}^n$ and $v_k^i\in\mathbb{R}^{m_i}, i = 1, 2, \ldots, N$ are independent zero-mean Gaussian noises with $\mathbb{E}[w_k w_l^T]=\delta_{kl}Q\ (Q\geq 0)$, $\mathbb{E}[v_k^i(v_l^j)^T]=\delta_{ij}\delta_{kl}R_i\ (R_i> 0)\quad \forall i,j\in\{1,2,\ldots,N\}$, and $\mathbb{E}[w_k (v_l^i)^T]=0\quad\forall i\in\in\{1,2,\ldots,N\},\quad k,l\in\mathbb{N}$. The initial state $x_0$ is Gaussian with mean $\bar{x}_0$ and covariance $P_0$. The pairs $(A,C_i), i= 1,2,\ldots,N$ are assumed to be observable and $(A,Q^{\frac{1}{2}})$ is controllable.

\subsection{Transmission Model}	
	In this system, the remote estimator and the sensors are equipped with a single antenna. Sensor $i$ can send its measurement $y_k^i$ to the remote estimator with transmission power $p_k^i\in[0,p^i_\mathrm{max}]$ over the fading channel denoted by $h_k^i$, which is modeled as
	\begin{equation}
		h_k^i = (T_k^i)^\frac{1}{2}g_k^i,
	\end{equation}
	where $p^i_\mathrm{max}$ is the maximum transmission power of sensor $i$, $g_k^i$ denotes the small scale fading at time $k$ and $T_k^i$ denotes the large scale fading, which is modeled as 
	\begin{equation}
		T_k^i = c_i f_k^i l_i,
	\end{equation}
	where $c_i$ is the path gain constant, $f_k^i$ is the shadow fading at time $k$, and $l_i$ is the pathloss.
	
	For sensor $i$, the signals transmitted by other sensors will at time $k$ be regarded as interferences. Consequently, the signal-to-interference-and-noise-ratio (SINR) of the signal transmitted by sensor $i$ at time $k$ is defined as
	\begin{equation}
		\mathrm{SINR}_k^i = \frac{h_k^i p_k^i}{\sum_{j\ne i}h_k^j p_k^j + \sigma^2},
	\end{equation}
	where $\sigma^2$ is the channel noise.
	
	To successfully decode the signal transmitted by sensor $i$, the following condition should be satisfied:
	\begin{equation}\label{qos}
	\mathrm{SINR}_k^i\geq \theta_i,
	\end{equation}
	where $\theta_i$ is the SINR quality-of-service (QoS) requirement for sensor $i$. Note that the condition (\ref{qos}) is widely used to guarantee successful data decoding in the wireless communication field~\cite{zhao2017joint}. 

	Due to the limited communication resources and the inter-inference between the sensors, the condition (\ref{qos}) may not be satisfied for each sensor. Therefore, at time $k$, it should be decided that which sensors should send their measurements. Define the transmission variable $\gamma_k^i$ as
	\begin{equation}
	\gamma_k^i = \left\{
	\begin{aligned}
		&1,\quad  \text{if $y_k^i$ is sent by sensor $i$ at time $k$,} \\
		&0, \quad \text{otherwise.}
		\end{aligned}
	\right.
	\end{equation}
	Furthermore, if sensor $i$ is selected, it should satisfy the condition (\ref{qos}) to guarantee a successful transmission; otherwise, there will be no constraint on the SINR of sensor $i$, i.e., the following condition should be satisfied:
	\begin{equation}
		\mathrm{SINR}_k^i\geq \gamma_k^i \theta_i.
	\end{equation}

  {
    \begin{assumption}\label{asm:feasibility}
        At each time step $k$, there exists at least one sensor $i$ such that
        \begin{equation*}
            \frac{h_k^i p^i_\mathrm{max}}{\sigma^2}\geq \theta_i.
        \end{equation*}
    \end{assumption}
    \begin{remark}
        Assumption~\ref{asm:feasibility} means there is at least one sensor's QoS requirement can be satisfied. Otherwise, none of the sensors can be selected, i.e., $\gamma_k^i = 0, \forall i$.
    \end{remark}
    }
\subsection{Estimation Process}
	After receiving the measurements from the sensors, the remote estimator runs a Kalman filter to obtain the minimum mean-squared error (MMSE) estimate of the system state $x_k$. Define $\boldsymbol{y}_k\triangleq (\gamma_k^1 (y_k^1)^T, \ldots, \gamma_k^N (y_k^N)^T)^T$	and let $\mathcal{L}_k\triangleq \{\boldsymbol{y}_1, \ldots, \boldsymbol{y}_k\}$. Then, define the priori estimate of {$x_k$} and its corresponding error covariance as follows:
	\begin{align}
		\hat{x}_{k|k-1}&\triangleq \mathbb{E}[x_k|\mathcal{L}_{k-1}],\\
		P_{k|k-1}&\triangleq \mathbb{E}[(x_k-\hat{x}_{k|k-1})(x_k-\hat{x}_{k|k-1})^T|\mathcal{L}_{k-1}],
	\end{align}
	and define the posteriori estimate of $x_k$ and its corresponding error covariance as follows:
	\begin{align}
		\hat{x}_{k|k}&\triangleq \mathbb{E}[x_k|\mathcal{L}_{k}],\\
		P_{k|k}&\triangleq \mathbb{E}[(x_k-\hat{x}_{k|k})(x_k-\hat{x}_{k|k})^T|\mathcal{L}_{k}].
	\end{align}
	The priori and posteriori estimates follow the below update steps:
	\begin{align}
		\hat{x}_{k|k-1}&=A\hat{x}_{k-1|k-1},\\
		P_{k|k-1}&=AP_{k-1|k-1}A^T+Q,\\
		\hat{x}_{k|k}&=\hat{x}_{k|k-1}+K_k(\boldsymbol{y}_k-\tilde{C}_k\hat{x}_{k|k-1}),\\
		P_{k|k}&=P_{k|k-1} - K_k\tilde{C}P_{k|k-1},\\
		K_k&=P_{k|k-1}\tilde{C}_k^T(\tilde{C}_k P_{k|k-1}\tilde{C}_k^T+\tilde{R}_k)^{\dagger},
	\end{align}
	where
	\begin{align*}
		\tilde{C}_k&\triangleq (\gamma_k^1 C_1^T,\ldots,\gamma_k^N C_N^T)^T,\\
		\tilde{R}_k&\triangleq \mathrm{diag}\{\gamma_k^1 R_1,\ldots,\gamma_k^N R_N\},
	\end{align*}
	and $\dagger$ represents the Moore-Penrose pseudo-inverse. 
    In the subsequent analysis, we will write $P_{k|k}$ as $P_k$ for simplicity whenever there is no confusion.
	{
	\begin{remark}
	    Note that the matrix $\tilde{C}_k P_{k|k-1}\tilde{C}_k^T+\tilde{R}_k$ is invertible if and only if $\gamma_k^i = 1, \forall i$. If not all sensors are selected at time $k$, i.e., $\exists i, \text{ s.t. } \gamma_k^i = 0$, the matrix $\tilde{C}_k P_{k|k-1}\tilde{C}_k^T+\tilde{R}_k$ will become a singular matrix. Since the Moore-Penrose pseudoinverse of an invertible matrix is its inverse, here we can use Moore-Penrose pseudoinverse to cover all cases~\cite{huang2021joint}.
	\end{remark}
	}
\subsection{Problem Description}
	In the following, we will formulate an optimization problem, where the transmission power of sensors and the selection variables are jointly optimized, to help select which sensors should transmit their measurements at time $k$. 
    {The set of candidate sensors $\mathcal{S}$ is initialized as $\{1,2,\ldots,N\}$, but some sensors should be remove from $\mathcal{S}$ due to the restriction of communication resources.}
	
    In practice, sensors usually have limited transmission power, i.e.,
	\begin{equation}
		p_k^i \in [0, p_\mathrm{max}^i],\quad \forall i\in\mathcal{S},
	\end{equation}
    The selection variables $\gamma_k^i (i=1,2,\ldots,N)$ are $0-1$ binary decision variables, i.e.,
	\begin{equation}
		\gamma_k^i\in\{0,1\},\quad\forall i\in\mathcal{S}.
	\end{equation}
	If a sensor is selected to transmit its measurement at time $k$, the corresponding SINR QoS requirement must be satisfied so that the transmitted measurement can be received by the remote estimator successfully, i.e., the following constraints should be satisfied:
	\begin{equation}\label{qos_constraint}
		\frac{h_k^i p_k^i}{\sum_{j\ne i}h_k^j p_k^j + \sigma^2}\geq \gamma_k^i \theta_i,\quad \forall i\in\mathcal{S}.
	\end{equation}
	
	{\begin{remark}
	    By utilizing modulation and coding, error-free decoding can be achieved if and only if the SINR at the receiver side is above a threshold, which is determined by the adopted modulation and coding schemes~\cite{dobre2011second}. Such a predefined threshold is called QoS requirement~\cite{li2007real}.
	\end{remark}
	}
	
	At time $k$, the system objective is to minimize the estimation error, which can be characterized as minimizing the trace of the error covariance of the state estimate $P_k$. Therefore, we have the following optimization problem:
	\begin{problem}\label{origin_problem}
		{\begin{align*}
		\min_{\boldsymbol{\gamma}_k, \boldsymbol{p}_k}\quad &\mathrm{Tr}\{P_k\}\\
		\mathrm{s.t.}\quad
		& \frac{h_k^i p_k^i}{\sum_{j\ne i}h_k^j p_k^j + \sigma^2}\geq \gamma_k^i \theta_i,\quad \forall i\in\mathcal{S},\\
		&p_k^i \in [0, p_\mathrm{max}^i],\quad \forall i\in\mathcal{S},\\
		&\gamma_k^i\in\{0,1\},\quad\forall i\in\mathcal{S},\\
        & P_{k}=P_{k|k-1} - K_k\tilde{C}_kP_{k|k-1},\\
		& K_k=P_{k|k-1}\tilde{C}_k^T\left(\tilde{C}_k P_{k|k-1}\tilde{C}_k^T+\tilde{R}\right)^{\dagger},
		\end{align*}}
        {where $\boldsymbol{\gamma}_k\triangleq \left[\gamma_k^1, \gamma_k^2, \ldots, \gamma_k^N \right]$, $\boldsymbol{p}_k \triangleq\left[p_k^1, p_k^2, \ldots, p_k^N \right]$, and $P_{k|k-1}=AP_{k-1}A^T+Q$. Moreover, $P_{k-1}$ and $h_k^i, \forall i$ are known by the estimator at time $k$. }
	\end{problem}
     {
     \begin{remark}
        In practice, the channel state information (CSI) $h_k^i, \forall i$  can be obtained by the estimator through channel estimation~\cite{coleri2002channel}, which can be accomplished by transmitting pilot sequences from each sensor.
     \end{remark}
     }                          
\section{Main Results}\label{sec:problem_solving}
	Due to the binary and non-convex constraints, it is extremely difficult to derive the optimal solution of Problem~\ref{origin_problem}. In this section, we will propose an efficient algorithm, which can obtain a suboptimal solution of Problem~\ref{origin_problem}. Fig.~\ref{relationship} illustrates the relationship among the main results derived in this section.
	\begin{figure}[h]
		\centering
		\includegraphics[width=3.4in]{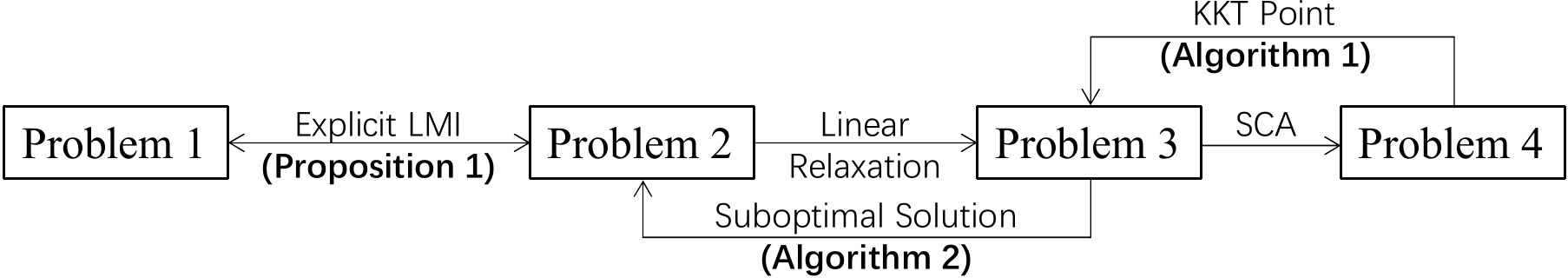}
		\caption{{Relationship among the main results}}
		\label{relationship}
	\end{figure}
\subsection{Problem Reformulation and Relaxation}
	We first eliminate the non-convexity of $P_{k}=P_{k|k-1} - K_k\tilde{C}P_{k|k-1}$ in the following Proposition~\ref{prop:1}:
	\begin{proposition}\label{prop:1}
		Problem~\ref{origin_problem} is equivalent to the following problem:
		\begin{problem}\label{nonconvex_problem}
			\begin{align*}
			\min_{\boldsymbol{\gamma}_k, \boldsymbol{p}_k, P_k} &\mathrm{Tr}\{P_k\}\\
			\mathrm{s.t.}\quad
			&\begin{bmatrix}
			(AP_{k-1}A^T+Q)^{-1} + \sum_{i=1}^{N}\gamma_k^i C_i^TR_i^{-1}C_i & I\\
			I & P_k
			\end{bmatrix}\\&\geq 0,\\
			& \frac{h_k^i p_k^i}{\sum_{j\ne i}h_k^j p_k^j + \sigma^2}\geq \gamma_k^i \theta_i,\quad \forall i\in\mathcal{S},\\
			&p_k^i \in [0, p_\mathrm{max}^i],\quad \forall i\in\mathcal{S},\\
			&\gamma_k^i\in\{0,1\},\quad\forall i\in\mathcal{S}.
			\end{align*}
		\end{problem}
	\end{proposition}
\begin{proof}
    {
    First of all, we introduce two lemmas:
    \begin{lemma}[Matrix Inversion Lemma]\label{lm:inversion}
        Let $A$ and $D$ be square, invertible matrices of size $n_A\times n_A$ and $n_D \times n_D$, and $B$ and $C$ be matrices of size $n_A\times n_D$ and $n_D\times n_A$, the following equation holds:
        \begin{equation*}
            (A+BD^{-1}C)^{-1} = A^{-1} - A^{-1}B(D-CA^{-1}B)^{-1}CA^{-1}.
        \end{equation*}
    \end{lemma}
    \begin{lemma}[Schur Complement]\label{lm:schur}
            Let $S$ be a symmetric matrix with the following formation
            $$S = 
            \begin{bmatrix}
                A & B \\
                B^T & C
            \end{bmatrix}.$$
            Then, if $A$ is positive definite, we have
            $$S\geq 0 \iff C - B^T A^{-1}B \geq 0.$$
        \end{lemma}
    Define 
    \begin{align*}
        \hat{C}_k &= \left(C_{i_1}^T, C_{i_2}^T, \ldots, C_{i_m}^T \right)^T,\\
        \hat{R}_k &= \mathrm{diag}\{ R_{i_1}, R_{i_2}, \ldots, R_{i_{m_k}} \},
    \end{align*}
    where $i_m\in\{i\in\mathcal{S}|\gamma_k^i = 1\}, \forall m\in\{1,2,\ldots, m_k\}$. 
    According to the definition of Moore-Penrose pseudo-inverse, we have
    \begin{equation}
    \begin{split}
         \tilde{C}_k^T&\left(\tilde{C}_k P_{k|k-1}\tilde{C}_k^T+\tilde{R}_k\right)^{\dagger}\tilde{C}_k \\&= \hat{C}_k^T\left(\hat{C}_k P_{k|k-1}\hat{C}_k^T+\hat{R}_k\right)^{-1}\hat{C}_k,
    \end{split}
    \end{equation}
    and hence we have
    \begin{equation}
        \begin{split}
        P_{k}&=P_{k|k-1} - K_k\tilde{C}_k P_{k|k-1}\\
        &= P_{k|k-1} - P_{k|k-1} \tilde{C}_k^T\left(\tilde{C}_k P_{k|k-1}\tilde{C}_k^T+\tilde{R}_k\right)^{\dagger} \tilde{C}_kP_{k|k-1}\\
        &= P_{k|k-1} - P_{k|k-1} \hat{C}_k^T\left(\hat{C}_k P_{k|k-1}\hat{C}_k^T+\hat{R}_k\right)^{-1} \hat{C}_kP_{k|k-1}\\
        &= \left(P_{k|k-1}^{-1}+ \hat{C}_k^T\hat{R}_k^{-1}\hat{C}_k\right)^{-1}\\
        &= \left(\left(AP_{k-1}A^T+Q\right)^{-1}+\sum_{i=1}^{N}\gamma_k^i C_i^TR_i^{-1}C_i\right)^{-1},
        \end{split}
    \end{equation}
        where the fourth equality follows from Matrix Inversion Lemma. 
        Then, Problem~\ref{origin_problem} can be rewritten as 
        \begin{align*}
		\min_{\boldsymbol{\gamma}_k, \boldsymbol{p}_k} \quad &\mathrm{Tr}\{P_k\}\\
		\mathrm{s.t.}\quad& P_k = \left((AP_{k-1}A^T+Q)^{-1}+\sum_{i=1}^{N}\gamma_k^i C_i^TR_i^{-1}C_i\right)^{-1},\\
		& \frac{h_k^i p_k^i}{\sum_{j\ne i}h_k^j p_k^j + \sigma^2}\geq \gamma_k^i \theta_i,\quad \forall i\in\mathcal{S},\\
		&p_k^i \in [0, p_\mathrm{max}^i],\quad \forall i\in\mathcal{S},\\
		&\gamma_k^i\in\{0,1\},\quad\forall i\in\mathcal{S}.
		\end{align*}
        which is equivalent to the following problem:
        \item[\quad\textbf{Problem~2.1:}] 
            \begin{align*}
		\min_{\boldsymbol{\gamma}_k, \boldsymbol{p}_k, P_k} &\mathrm{Tr}\{P_k\}\\
		\mathrm{s.t.}\quad& P_k \geq \left((AP_{k-1}A^T+Q)^{-1}+\sum_{i=1}^{N}\gamma_k^i C_i^TR_i^{-1}C_i\right)^{-1},\\
		& \frac{h_k^i p_k^i}{\sum_{j\ne i}h_k^j p_k^j + \sigma^2}\geq \gamma_k^i \theta_i,\quad \forall i\in\mathcal{S},\\
		&p_k^i \in [0, p_\mathrm{max}^i],\quad \forall i\in\mathcal{S},\\
		&\gamma_k^i\in\{0,1\},\quad\forall i\in\mathcal{S}.
		\end{align*} 	
        The equivalence is derived from the fact that the constraint $P_k \geq \left((AP_{k-1}A^T+Q)^{-1}+\sum_{i=1}^{N}\gamma_k^i C_i^TR_i^{-1}C_i\right)^{-1}$ in Problem~2.1 must be satisfied with equality. Otherwise, we can always decrease the objective value by 
        replacing $P_k$ with $\left((AP_{k-1}A^T+Q)^{-1}+\sum_{i=1}^{N}\gamma_k^i C_i^TR_i^{-1}C_i\right)^{-1}$.
        }
        
        {
        Due to $Q\geq 0$, $R_i> 0, \forall i$, and Assumption~\ref{asm:feasibility}, we have
        \begin{equation}
        \begin{split}
           (AP_{k-1}A^T+Q)^{-1}+\sum_{i=1}^{N}\gamma_k^i C_i^TR_i^{-1}C_i >0.
        \end{split}
        \end{equation}
        Similar to~\cite{8920188}, by Lemma~\ref{lm:schur}, we have $P_k \geq \left((AP_{k-1}A^T+Q)^{-1}\right.$ $\left.+\sum_{i=1}^{N}\gamma_k^i C_i^TR_i^{-1}C_i\right)^{-1}$ is equivalent to
		\begin{equation}\label{LMI}
		\begin{bmatrix}
		(AP_{k-1}A^T+Q)^{-1} + \sum_{i=1}^{N}\gamma_k^i C_i^TR_i^{-1}C_i & I\\
		I & P_k
		\end{bmatrix}\geq 0.
		\end{equation}
		The proof is thus completed.
        }
\end{proof}

	To deal with the integer (binary) constraints, we adopt the linear relaxation, which is often considered in the literature~\cite{porta2009linear,shi2013optimal,huang2021joint}. Specifically, we replace the constraints $\gamma_k^i\in\{0,1\}, \forall i$ with $\gamma_k^i\in[0,1],\forall i$. Then, we have the following relaxed problem:
	\begin{problem}\label{relaxed_problem}
		\begin{align*}
		\min_{\boldsymbol{\gamma}_k, \boldsymbol{p}_k, P_k} &\mathrm{Tr}\{P_k\}\\
		\mathrm{s.t.}\quad&
		\begin{bmatrix}
		(AP_{k-1}A^T+Q)^{-1} + \sum_{i=1}^{N}\gamma_k^i C_i^TR_i^{-1}C_i & I\\
		I & P_k
		\end{bmatrix}\\&\geq 0,\\
		& \frac{h_k^i p_k^i}{\sum_{j\ne i}h_k^j p_k^j + \sigma^2}\geq \gamma_k^i \theta_i,\quad \forall i\in\mathcal{S},\\
		&p_k^i \in [0, p_\mathrm{max}^i],\quad \forall i\in\mathcal{S},\\
		&\gamma_k^i\in[0,1],\quad\forall i\in\mathcal{S}.
		\end{align*}
	\end{problem}

	\begin{remark}
		Note that after the linear relaxation for variables $\gamma_k^i$, the constraint (\ref{LMI}) is a linear matrix inequality (LMI), which is convex. However, the non-convexity of the constraint (\ref{qos_constraint}) has not been eliminated yet. Therefore, further efforts should be made to deal with this issue.
	\end{remark}

\subsection{Algorithm for Solving Problem~\ref{relaxed_problem}}
	As discussed above, Problem~\ref{relaxed_problem} is still non-convex due to the constraint (\ref{qos_constraint}). In this subsection, we propose an iterative algorithm based on SCA technique~\cite{park2015multihop, liu2017cross, liu2018transceiver}, where a series of convex problems will be solved to approximate the optimal solution of the non-convex Problem~\ref{relaxed_problem}.
	{In each iteration, a convex surrogate problem will be constructed based on the optimal solution obtained in the last iteration. Then, the new surrogate problem will be solved in the subsequent iteration. Eventually, the solution of the surrogate problem will converge to a solution satisfying the KKT conditions of Problem~\ref{relaxed_problem}. In the following, we will show how to construct such surrogate problems.}
	
	By introducing a set of auxiliary variables $\{\eta_k^i, \forall i\}$, the constraint (\ref{qos_constraint}) can be equivalently rewritten as the following two constraints:
	\begin{align}
		 \sum_{j\ne i}h_k^j p_k^j + \sigma^2&\leq \eta_k^i,\quad\forall i,\label{c1}\\
		\eta_k^i \gamma_k^i\theta_i&\leq h_k^i p_k^i,\quad\forall i.\label{c2}
	\end{align}
	The constraint (\ref{c1}) is a convex (linear) constraint, but the constraint (\ref{c2}) is still non-convex. 
	{From the fact $4\eta^i_k\gamma^i_k\leq 4\eta^i_k\gamma^i_k+(\eta^i_k-\gamma^i_k-(b^i_k)^{(t)})^2$, we can see that the non-convex term $\eta_k^i \gamma_k^i$ has the following convex bound:}
	\begin{equation}\label{eq:approximation}
		\eta_k^i \gamma_k^i\leq \frac{(\eta_k^i+\gamma_k^i)^2-2(b_k^i)^{(t)}(\eta_k^i-\gamma_k^i)+((b_k^i)^{(t)})^2}{4},
	\end{equation}
	where $(b_k^i)^{(t)}=(\eta_k^i)^{(t)}-(\gamma_k^i)^{(t)}$, and $(\eta_k^i)^{(t)}$ and $(\gamma_k^i)^{(t)}$ are the feasible points in the $t$-th iteration.
	Then Problem~\ref{relaxed_problem} can be approximated by the following convex surrogate problem in the {$(t+1)$}-th iteration:
	\begin{problem}\label{convex_problem}
		\begin{align*}
		\min_{\{\boldsymbol{\gamma}_k, \boldsymbol{\eta}_k, \boldsymbol{p}_k, \atop P_k\}^{(t+1)}} &\mathrm{Tr}\{P_k\}\\
		\mathrm{s.t.}\quad
		& \begin{bmatrix}
		(AP_{k-1}A^T+Q)^{-1} + \sum_{i=1}^{N}\gamma_k^i C_i^TR_i^{-1}C_i & I\\
		I & P_k
		\end{bmatrix}\\&\geq 0,\\
		& \sum_{j\ne i}h_k^j p_k^j + \sigma^2\leq \eta_k^i,\quad\forall i\in\mathcal{S},\\
		&\left(\eta_k^i+\gamma_k^i\right)^2-2{\left(b_k^i\right)^{(t)}}\left(\eta_k^i-\gamma_k^i\right)+\left( {\left(b_k^i\right)^{(t)}} \right)^2\\&\leq \frac{4}{\theta_i}h_k^i p_k^i, \quad\forall i\in\mathcal{S},\\
		&p_k^i \in [0, p_\mathrm{max}^i],\quad\forall i\in\mathcal{S},\\
		&\gamma_k^i\in[0,1],\quad\forall i\in\mathcal{S}.
		\end{align*}
	\end{problem}
	The SCA-based algorithm for solving Problem~\ref{relaxed_problem} is summarized in Algorithm~\ref{algo:sca}:
	\begin{algorithm}
			\caption{SCA-based algorithm for solving Problem~\ref{relaxed_problem}}
		\begin{algorithmic}\label{algo:sca}
			\STATE Set the initial values for $\{\boldsymbol{\gamma}_k,\boldsymbol{\eta}_k,\boldsymbol{p}_k,P_k\}^{(0)}$ and set {$t=0$};
			\REPEAT
			\STATE Update {$(b_k^i)^{(t)}=(\eta_k^i)^{(t)}-(\gamma_k^i)^{(t)}$};
                \STATE Solve Problem~\ref{convex_problem} to obtain {$\{\boldsymbol{\gamma}_k,\boldsymbol{\eta}_k,\boldsymbol{p}_k,P_k\}^{(t+1)}$}; 
			\STATE $t= t+1$;
			\UNTIL{converge}
		\end{algorithmic}
	\end{algorithm}
	\begin{proposition}\label{prop:convergence}
		Monotonic convergence of Algorithm~\ref{algo:sca} is guaranteed. Moreover, the converged solution satisfies all the constraints as well as the KKT conditions of Problem~\ref{relaxed_problem}.
	\end{proposition}
	\begin{proof}
		See Appendix.
	\end{proof}
	\begin{remark}\label{rmk:initialization}
		{The initial values of $\{\boldsymbol{\gamma}_k,\boldsymbol{\eta}_k,\boldsymbol{p}_k,P_k\}^{(0)}$ must be feasible for Problem~\ref{relaxed_problem}. We can simply set $\gamma^i_k=0, p^i_k=0, \eta^i_k=\sigma^2, \forall i$ and $P_k=AP_{k-1}A^T+Q$. Note that the initial values will influence the number of iterations. }
	\end{remark}
 \textcolor{black}{
    \begin{remark}
        Generally, a KKT point of a non-convex problem can be a global minimum, a local minimum, a saddle point, or even a maximum. 
        Usually, obtaining the global minimum of a non-convex problem is considered NP-hard. Therefore, the majority of efforts in solving non-convex problems aim to obtain a local minimum~\cite{tang2019user, yang2021joint}. 
        According to~\cite[Corollary 1]{marks1978general}, the solution obtained by Algorithm~\ref{algo:sca} is not only a KKT point of Problem~\ref{relaxed_problem}, but also a local minimum to Problem~\ref{relaxed_problem}.
        The local minimum obtained by Algorithm~\ref{algo:sca} may vary with the initial point. In Section~\ref{sec:simulation}, the results were derived with the initial values mentioned in Remark~\ref{rmk:initialization}. 
    \end{remark}
 }
\subsection{Heuristic for Sensor Selection}
	So far, Problem~\ref{relaxed_problem} has been solved by Algorithm~\ref{algo:sca}. However, the values of the transmission variables $\gamma_k^i, \forall i$ obtained by solving Problem~\ref{relaxed_problem} are continuous rather than binary due to the linear relaxation. That is to say, there is still a gap before we obtain the solution to the original problem, i.e., Problem~\ref{origin_problem}. 
	
	Based on the above analysis, there is no doubt that it is difficult to find a feasible solution to Problem~\ref{origin_problem}, let alone to find the global optimal solution. Therefore, in this subsection, we will provide a heuristic to give a suboptimal solution to Problem~\ref{origin_problem}.
	
	Define two functions $h:\mathbb{S}_n^{+}\rightarrow\mathbb{S}_n^{+}$ and $g(X;S):\mathbb{S}_n^{+}\times\mathbb{S}_n^{+}\rightarrow\mathbb{S}_n^{+}$:
	\begin{align}
		h(X) &\triangleq AXA^T+Q,\\
		g(X;S) &\triangleq ([h(X)]^{-1}+S)^{-1}.
	\end{align}
	It is easy to see that $g(X;S)$ is matrix monotonically decreasing with respect to $S$ in $\mathbb{S}_n^{+}$. 
    By defining $\gamma_k^i C_i^T R_i^{-1} C_i$ as the \textit{assimilated sensing precision matrix} for sensor $i$, based on the fact that $P_k=g\left(P_{k-1},\sum_{i=1}^{N}\gamma_k^i C_i^T R_i^{-1} C_i\right)$, we can heuristically select the sensors which have ``larger'' assimilated sensing precision matrices. 
    {In the following, we take the trace of the assimilated sensing precision matrix, i.e., $\mathrm{Tr}\{\gamma_k^i C_i^T R_i^{-1} C_i\}$, as the selection metric. We say that the assimilated sensing precision matrix of sensor $i$ is ``larger''  than that of sensor $j$ if $\mathrm{Tr}\{\gamma_k^i C_i^T R_i^{-1} C_i\} > \mathrm{Tr}\{\gamma_k^j C_j^T R_j^{-1} C_j\}$.}
    
        The intuition of using $\mathrm{Tr}\{\gamma_k^i C_i^T R_i^{-1} C_i\}$ as the selection metric is to jointly consider how far a sensor's QoS requirement can be satisfied and how much the sensor can contribute to the estimation accuracy. The extent to which the QoS requirement can be satisfied can be characterized by $\gamma^i_k$. 
        If $\gamma^i_k = 1$, then the QoS requirement of sensor $i$ can be fully satisfied. 
        If $\gamma^i_k < 1$, a larger value of $\gamma^i_k$ indicates that the QoS requirement of sensor $i$ is more likely to be satisfied, which further implies that sensor $i$ is more capable of transmission. 
        The contribution to the estimation accuracy of sensor $i$ can be characterized by the matrix $C_i^T R_i^{-1} C_i$. Specifically, if $C_i^T R_i^{-1} C_i>C_j^T R_j^{-1} C_j (i\ne j)$, then sensor $i$ will be credited with a greater contribution to the estimation accuracy than sensor $j$. Therefore, to combine these two considerations as well as to facilitate calculations and numerical comparisons, the trace of the assimilated sensing precision matrix is adopted as the selection metric. 
    
	The whole algorithm for sensor selection is summarized in Algorithm~\ref{algo:sensor_selection}. 
    In each iteration, the sensor with the smallest value of the trace of the assimilated sensing precision matrix will be removed from the candidate set $\mathcal{S}$. 
    \textcolor{black}{Note that if the QoS requirements of all the sensors in $\mathcal{S}$ can be satisfied with the communication resources of the system, then $\gamma_k^i (\forall i\in\mathcal{S})$ will be $1$ so that the objective value, i.e., $\mathrm{Tr}\{P_k\}$, can be minimized subject to the communication resource constraints. That is, if there exists some $\boldsymbol{\eta}_k$, $\boldsymbol{p}_k$ and $P_k$ such that $\{\boldsymbol{\gamma}_k,\boldsymbol{\eta}_k,\boldsymbol{p}_k,P_k\}$, where $\boldsymbol{\gamma}_k=[1,1,\ldots, 1]$, is a feasible solution for Problem~\ref{convex_problem}, then it must be the optimal solution. Due to interference between the sensors, the system communication resources may not be enough to accommodate the transmission of all the sensors in $\mathcal{S}$ under the QoS constraints. As a result, there will exist $i\in\mathcal{S}$ such that $\gamma_k^i<1$. As more and more sensors are removed from the candidate set $\mathcal{S}$, the interference between sensors can be reduced such that the QoS requirements of the remaining sensors can be satisfied. Gradually, $\gamma_k^i\in\mathcal{S}$ will approach $1$. }
    Therefore, the algorithm will end when $\gamma_k^i = 1, \forall i\in\mathcal{S}$, and the sensors in the candidate set $\mathcal{S}$ will be selected at time $k$.  
	\begin{algorithm}
		\caption{Sensor Selection Algorithm}
		\begin{algorithmic}\label{algo:sensor_selection}
			\STATE Initialize the set of candidate sensors $\mathcal{S}=\{1,2,\ldots,N\}$;
			\LOOP
			\STATE Solve Problem~\ref{relaxed_problem} for the system including the sensors in $\mathcal{S}$ via Algorithm~\ref{algo:sca}; 
			\IF {$\gamma_k^i = 1, \forall i\in\mathcal{S}$}
			\RETURN $\mathcal{S}$.
			\ELSE
			\STATE  $s=\arg\min_{ i\in\mathcal{S}}\left\{\mathrm{Tr}\{\gamma_k^i C_i^T R_i^{-1} C_i\}\right\}$;
			\STATE $\mathcal{S}=\mathcal{S}\backslash s$;
			\ENDIF
			\ENDLOOP
		\end{algorithmic}
	\end{algorithm}

\section{Simulation Results}\label{sec:simulation}
In the following simulations, all sensors are randomly deployed in a circular area with a radius of 2km, and the remote estimator is located at the center of the circle. The shadow fading is modeled as a random variable following log-normal distribution with zero mean and 8dB standard deviation. The path gain constant $c_i (i=1,2,\ldots,N)$ are set to 1. The pathloss is modeled as $l_i(d_i)=-147.3-43.3\log(d_i)$dB, where $d_i$ is the distance between the remote estimator and sensor $i$. The small-scale fading is modeled as Rayleigh fading with zero mean and unit variance. We set the noise power $\sigma^2=-30\textrm{dB}$ and the maximum transmission power $p_\textrm{max}^i = 1\textrm{mW}, \forall i$.
{According to~\cite{tsiropoulos2014radio}, the attainable data rate is monotonically non-decreasing with the SINR. Specifically, given the SINR QoS requirement $\theta_i$, the minimum attainable data rate $r_i$ of sensor $i$ is $B\log_2(1+\theta_i)$, where $B$ is the bandwidth. In the following simulations, all sensors are assumed to have the same minimum attainable data rate $r=50\textrm{Mbps}$, and hence we have $\theta_i=2^{\frac{r}{B}}-1, \forall i$.}

\subsection{Convergence of Algorithm~\ref{algo:sca}}
	First, we verify the convergence of Algorithm~\ref{algo:sca} proposed in Section~\ref{sec:problem_solving}. 
	In this simulation, we consider an unstable system with the $A=1.1$ and $Q=1$. There are $N=10$ sensors in the system, and $C_i\in\mathbb{R} (i=1,2,\ldots,N)$ and $R_i\in\mathbb{R} (i=1,2,\ldots,N)$ are chosen randomly.
	
	Fig.~\ref{convergence} shows the convergence behavior of Algorithm~\ref{algo:sca}. One can see that Algorithm~\ref{algo:sca} is monotonically convergent, which confirms Proposition~\ref{prop:convergence}.
\begin{figure}[h]
	\centering
	\includegraphics[width=3.5in]{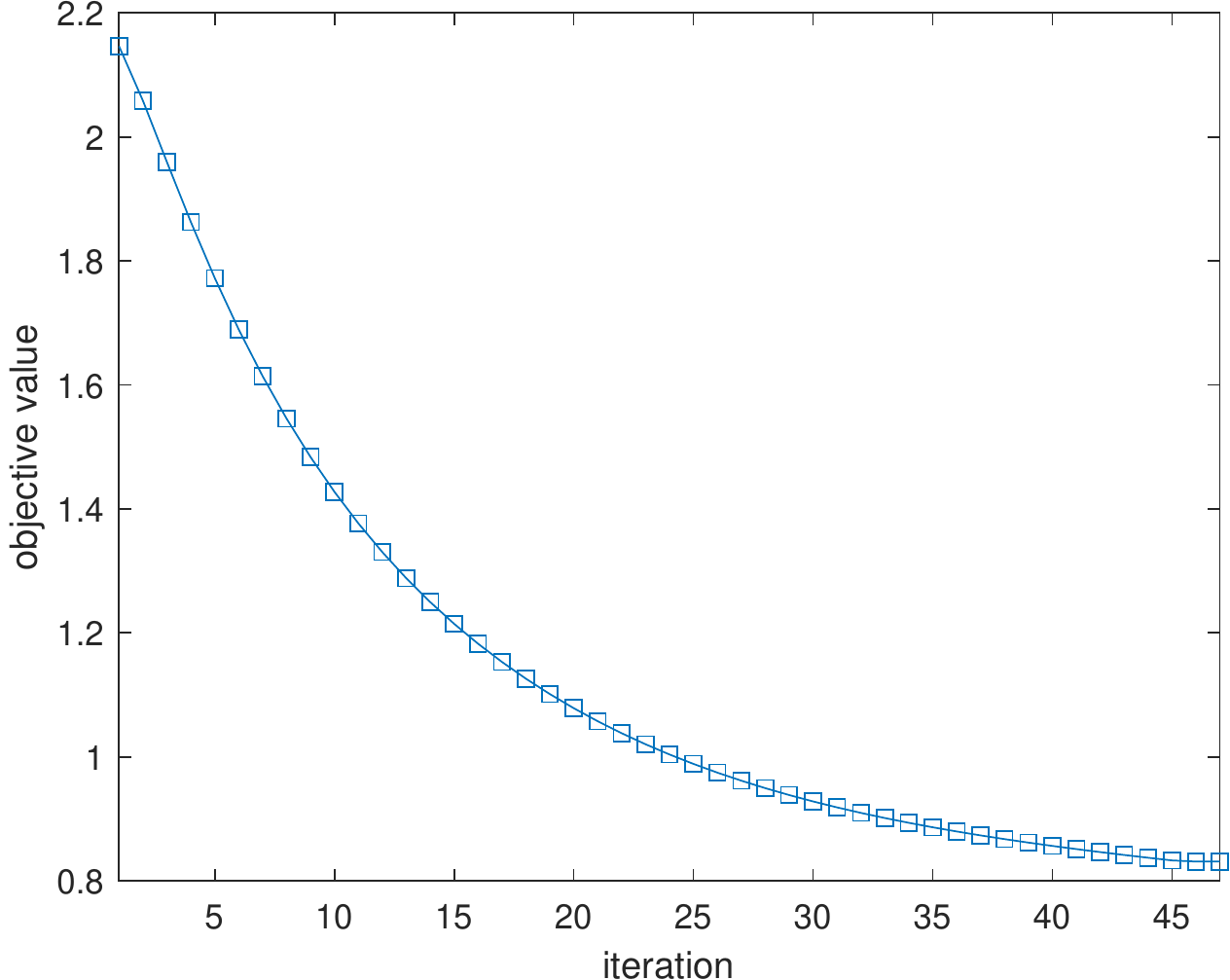}
	\caption{Convergence behavior of Algorithm~\ref{algo:sca}}
	\label{convergence}
\end{figure}

\subsection{Effectiveness of Algorithm~\ref{algo:sensor_selection}}
	In this subsection, we show the proposed heuristic (Algorithm~\ref{algo:sensor_selection}) is effective by comparing it with the following heuristic sensor selection methods:
	\begin{itemize}
		\item \textbf{Sensor number maximization (SNM)\cite{tang2019user}}: the main idea of this method is to maximize the number of sensors which the system can support with the limited communication resources. This method is often used in many communication scenarios for user selection.
		\item {\textbf{Precise measurements first (PMF)}: the basic idea of this heuristic is to select as many sensors with more precise measurements as possible, while ensuring that the QoS requirement constraints are satisfied. This method follows the below steps:
	\begin{itemize}
		\item[1)] Initialize the sensor set $\mathcal{S}=\{1,2,\ldots,N\}$ and the candidate set $\mathcal{C}=\emptyset$;
		\item[2)] Calculate $s=\arg\max_{i\in\mathcal{S}}\{\mathrm{Tr}\{C_i^TR_i^{-1}C_i\} \}$, and let $\mathcal{S}=\mathcal{S}\backslash s$;
		\item[3)] Check the feasibility of the following optimization problem:
		\begin{equation}\label{problem_PMF}
			\begin{split}
				\min_{{p}^i_k, i\in\mathcal{C}\cup\{s\}}& c\\
				\mathrm{s.t.}\quad& \frac{h_k^i p_k^i}{\sum_{j\in\mathcal{C}\cup\{s\}\backslash i}h_k^j p_k^j + \sigma^2}\geq \theta_i, \forall i\in\mathcal{C}\cup\{s\},
			\end{split}
		\end{equation}
		where $c$ can be any constant. If the problem is feasible, then let $\mathcal{C}=\mathcal{C}\cup \{s\}$;
		\item[4)] If $\mathcal{S}\ne\emptyset$, Go back to step 2); otherwise, output $\mathcal{C}$.
	\end{itemize}
}
	\end{itemize}
	{We consider a system with $x_k\in\mathbb{R}^5$, $y_k^{i}\in\mathbb{R}^{5},\forall i$, and the following parameters:
	\begin{align*}
		A&=\begin{bmatrix}
		 0.9416 & -0.0180 &  0.0715 &  0.0262 & -0.0196\\
		-0.0559 &  0.9948 &  0.0544 &  0.0251 & -0.0148\\
		-0.0564 & -0.0176 &  1.0686 &  0.0255 & -0.0162\\
		-0.0631 & -0.0090 &  0.0652 &  1.0284 & -0.0179\\
		-0.0197 & -0.0046 &  0.0200 &  0.0096 &  1.0049\\
		\end{bmatrix},\\
		Q &= I.
	\end{align*}
	There are $N=10$ sensors in the system. The value of $\mathrm{Tr}\{P_k\}$ is obtained by averaging $1000$ independent experimental results with different sensor location, channel states, matrices $C_i(i=1,2,\ldots,N)$ and $R_i\leq 5I(i=1,2,\ldots,N)$ generated randomly. }
	\begin{figure}[h]
		\centering
		\includegraphics[width=3.5in]{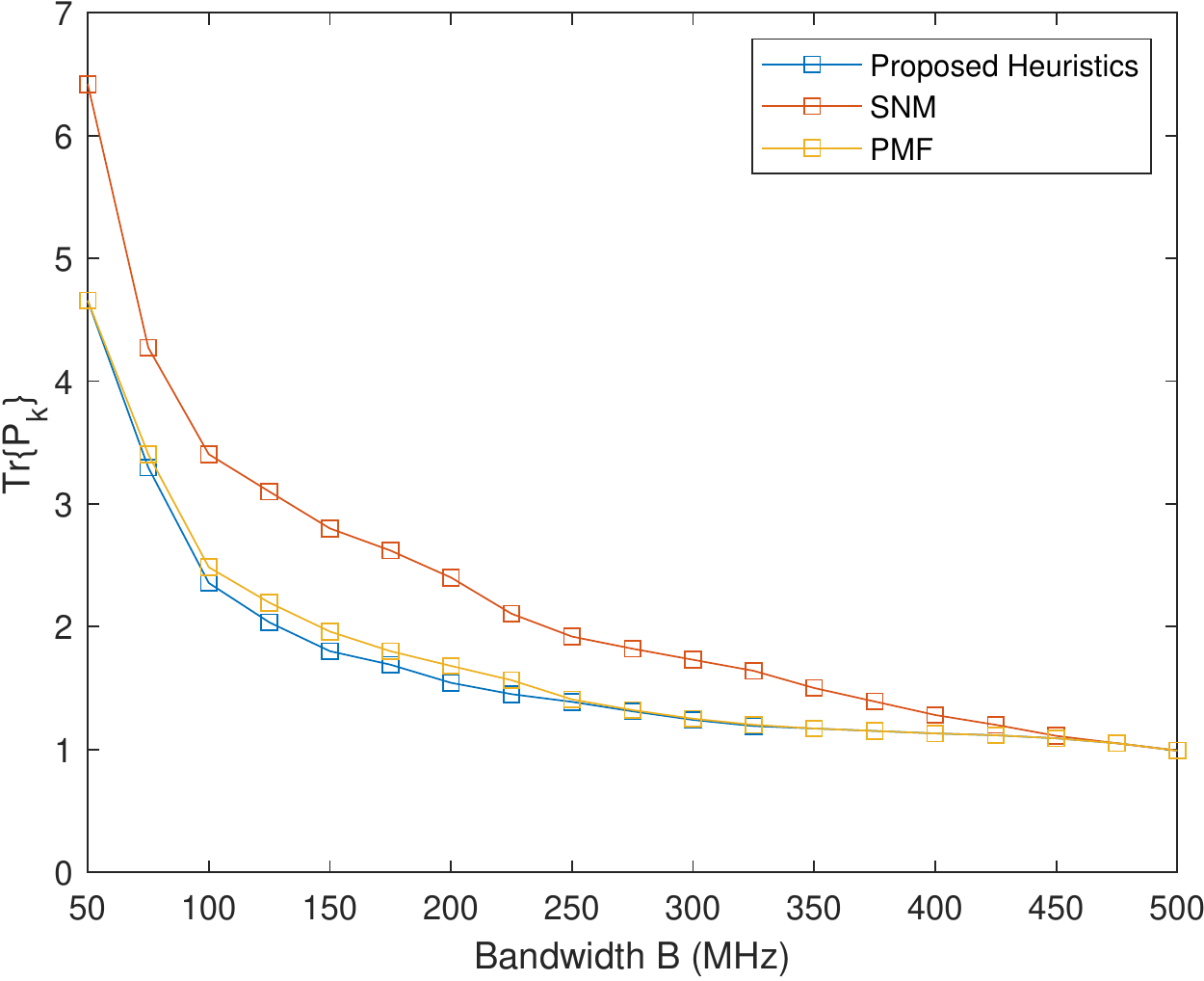}
		\caption{{Trace of error covariance matrix $P_k$ versus bandwidth $B$}}
		\label{high_dimensional_system}
	\end{figure}
	Fig.~\ref{high_dimensional_system} shows the trace of error covariance $P_k$ versus bandwidth $B$. One can see that our proposed heuristic always performs better than the SNM method and the PMF method, especially when there is less available bandwidth. 
    When bandwidth resources are enough to allow all the sensors to transmit their measurements simultaneously, the proposed method and the other two methods have the same performance since all the sensors will be selected. When bandwidth resources decrease, the SNM method tends to choose the sensors which have better channel states so that more sensors can transmit their measurements. However, it is inadequate to only consider selecting as many sensors as possible, since the sensing precision of a sensor also needs to be taken into account. At the extreme, the sensor which has the best channel state may have the worst sensing precision. As a result, this sensor may not deserve to be selected although it introduces low communication cost. On the contrary, a sensor which has little worse channel state but better sensing precision can be a good candidate to be selected. 
    The PMF method performs as well as the proposed method when there is more bandwidth available, which is because the system has more freedom to select the sensors with the most accurate measurement as much as possible. When the available bandwidth is limited, it is also inadequate to consider the measurement accuracy only. 
	Since our proposed method takes both transmission cost and measurement precision of sensors into account, it outperforms the SNM method and the PMF method especially in the low bandwidth region.
 
    {Moreover, when the bandwidth is so low that the system can only support one sensor to transmit, the optimal solution is to select the sensor with the most accurate measurement. In this case, the PMF method can obtain the optimal solution. The simulation results show that the proposed heuristic tends to have similar performance with the PMF when the bandwidth decreases, which implies that the solution obtained by the proposed heuristic is approaching the optimal solution.}

{
\subsubsection{Case Study}
In the following simulations, we study two simple cases to see how Algorithm~\ref{algo:sensor_selection} works.
}
\begin{table*}[h]\label{tbl:1}
\color{black}
\centering
\caption{Solving process of Algorithm~\ref{algo:sensor_selection} (case~1)}
\setlength{\tabcolsep}{3.5mm}{
\begin{tabular}{|cl|cl|cl|cl|cl|}
\hline
\multicolumn{2}{|c|}{Iteration} & \multicolumn{2}{c|}{$\mathcal{S}$}   & \multicolumn{2}{c|}{$\boldsymbol{\gamma}_k$}                      & \multicolumn{2}{c|}{$\boldsymbol{p}_k$}                           & \multicolumn{2}{c|}{$s$} \\ \hline
\multicolumn{2}{|c|}{1}         & \multicolumn{2}{c|}{$\{1,2,3,4,5\}$} & \multicolumn{2}{c|}{$\{0.1882, 1.0000, 0.0600, 1.0000, 1.0000\}$} & \multicolumn{2}{c|}{$\{0.0148, 0.1200, 1.0000, 0.1200, 0.1200\}$} & \multicolumn{2}{c|}{$1$} \\ \hline
\multicolumn{2}{|c|}{2}         & \multicolumn{2}{c|}{$\{2,3,4,5\}$}   & \multicolumn{2}{c|}{$\{1.0000, 0.1559, 1.0000, 1.0000\}$}         & \multicolumn{2}{c|}{$\{0.0483, 1.0000, 0.0483, 0.0483\}$}          & \multicolumn{2}{c|}{$3$} \\ \hline
\multicolumn{2}{|c|}{3}         & \multicolumn{2}{c|}{$\{2,4,5\}$}     & \multicolumn{2}{c|}{$\{1.0000, 1.0000, 1.0000\}$}                 & \multicolumn{2}{c|}{$\{0.7741, 0.7741, 0.7741\}$}                 & \multicolumn{2}{c|}{-}   \\ \hline
\end{tabular}
}
\end{table*}

\begin{table*}[h]\label{tbl:2}
\color{black}
\centering
\caption{Solving process of Algorithm~\ref{algo:sensor_selection} (case~2)}
\setlength{\tabcolsep}{3.5mm}{
\begin{tabular}{|cl|cl|cl|cl|cl|}
\hline
\multicolumn{2}{|c|}{Iteration} & \multicolumn{2}{c|}{$\mathcal{S}$}   & \multicolumn{2}{c|}{$\boldsymbol{\gamma}_k$}                      & \multicolumn{2}{c|}{$\boldsymbol{p}_k$}                          & \multicolumn{2}{c|}{$s$} \\ \hline
\multicolumn{2}{|c|}{1}         & \multicolumn{2}{c|}{$\{1,2,3,4,5\}$} & \multicolumn{2}{c|}{$\{1.0000, 0.1015, 1.0000, 0.1015, 0.1015\}$} & \multicolumn{2}{c|}{$\{0.0050, 0.0014, 1.0000, 0.0014, 0.0014\}$} & \multicolumn{2}{c|}{$2$} \\ \hline
\multicolumn{2}{|c|}{2}         & \multicolumn{2}{c|}{$\{1,3,4,5\}$}   & \multicolumn{2}{c|}{$\{1.0000, 1.0000, 0.1558, 0.1558\}$}         & \multicolumn{2}{c|}{$\{0.0050, 1.0000, 0.0021, 0.0021\}$}         & \multicolumn{2}{c|}{$4$} \\ \hline
\multicolumn{2}{|c|}{3}         & \multicolumn{2}{c|}{$\{1,3,5\}$}     & \multicolumn{2}{c|}{$\{1.0000, 1.0000, 0.3332\}$}                 & \multicolumn{2}{c|}{$\{0.0050, 1.0000, 0.0041\}$}                 & \multicolumn{2}{c|}{$5$} \\ \hline
\multicolumn{2}{|c|}{4}         & \multicolumn{2}{c|}{$\{1,3\}$}       & \multicolumn{2}{c|}{$\{1.0000, 1.0000\}$}                         & \multicolumn{2}{c|}{$\{0.0055, 0.9588\}$}                        & \multicolumn{2}{c|}{-}   \\ \hline
\end{tabular}
}
\end{table*}

\begin{table*}[h]\label{tbl:3}
    \color{black}
    \caption{Solutions obtained by different methods}
    \centering
    \setlength{\tabcolsep}{6mm}{
    \begin{tabular}{|cc|c|c|c|c|}
    \hline
    \multicolumn{2}{|c|}{\textbf{Method}}                                       & \textbf{SNM} & \textbf{PMF} & \textbf{Proposed Heuristic} & \multicolumn{1}{l|}{\textbf{Optimal}} \\ \hline
    \multicolumn{1}{|c|}{\multirow{2}{*}{\textbf{Case 1}}} & Selected Sensor    & $\{2,4,5\}$  & $\{2,3\}$    & $\{2,4,5\}$                 & $\{2,4,5\}$                                    \\ \cline{2-6} 
    \multicolumn{1}{|c|}{}                                 & Objective Function & $0.0645$     & 0.0822       & $0.0645$                    & $0.0645$                                       \\ \hline
    \multicolumn{1}{|c|}{\multirow{2}{*}{\textbf{Case 2}}} & Selected Sensor    & $\{2,4,5\}$  & $\{1,3\}$    & $\{1,3\}$                   & $\{1,3\}$                                      \\ \cline{2-6} 
    \multicolumn{1}{|c|}{}                                 & Objective Function & $0.2857$     & $0.1091$     & $0.1091$                    & $0.1091$                                       \\ \hline
    \end{tabular}
    }
\end{table*}

{
\emph{Case 1:} $C_i = 1, \forall i$, $R_1 = 0.5, R_2 = 0.2, R_3 = 0.15, R_4 = 0.2, R_5 = 0.2$, $h_k^1 = 2, h_k^2 = 1, h_k^3 = 0.01, h_k^4 = 1, h_k^5 = 1$. Hence, we have $\mathrm{Tr}\{C_3^TR_3^{-1}C_3\}>\mathrm{Tr}\{C_2^TR_2^{-1}C_2\}=\mathrm{Tr}\{C_4^TR_4^{-1}C_4\}=\mathrm{Tr}\{C_5^TR_5^{-1}C_5\}>\mathrm{Tr}\{C_1^TR_1^{-1}C_1\}$.}

{
\emph{Case 2:} $C_i = 1, \forall i$, $R_1 = 0.5, R_2 = 1, R_3 = 0.15, R_4 = 1, R_5 = 1$, $h_k^1 = 2, h_k^2 = 1, h_k^3 = 0.01, h_k^4 = 1, h_k^5 = 1$. Hence, we have $\mathrm{Tr}\{C_3^TR_3^{-1}C_3\}>\mathrm{Tr}\{C_1^TR_1^{-1}C_1\}>\mathrm{Tr}\{C_2^TR_2^{-1}C_2\}=\mathrm{Tr}\{C_4^TR_4^{-1}C_4\}=\mathrm{Tr}\{C_5^TR_5^{-1}C_5\}$.
}

{Moreover, in both cases, we set $A=1.005$, $Q = 1$, $P_{k-1} = 1$, $\sigma^2 = -20\textrm{dB}$, $p_\mathrm{max}^i = 1\textrm{mW}, \forall i$, and $\theta_i=\sqrt{2}-1, \forall i$. Note that the differences between case~1 and case~2 are the values of $R_2$, $R_4$, and $R_5$.}

{
Table~1 and Table~2 show the solving process of Algorithm~\ref{algo:sensor_selection} in case~1 and case~2, respectively. 
Table~3 shows the solutions obtained by different methods. In both case~1 and case~2, the set of selected sensors obtained by the SNM method is $\{2,4,5\}$, and the objective values are $0.0645$ and $0.2857$, respectively. Furthermore, the solution obtained by the SNM method will always maintain unchanged if $p_\mathrm{max}^i, \forall i$, $h_k^i, \forall i$ and $\theta^i, \forall i$ are fixed. The reason is that the SNM method only considers the constraints on communication resources, which are merely related to $p_\mathrm{max}^i, \forall i$, $h_k^i, \forall i$ and $\theta^i, \forall i$. Although the obtained solution is optimal for case~1, it cannot be optimal for most cases. 
The PMF method can obtain the optimal solution for case~2 but cannot for case~1. 
However, the proposed heuristic can obtain the optimal solutions for both case~1 and case~2, which further reflects the effectiveness and versatility of the method.
}

\section{Conclusion}
	In this paper, we study the sensor selection problem for remote state estimation under the QoS requirement constraints. 
	The formulated sensor selection problem is non-convex and it is difficult to obtain a feasible solution. To deal with this problem, we propose a heuristic algorithm, which is derived with the help of linear relaxation, SCA technique, and the concept of assimilated sensing precision matrix. Simulation results show that the proposed heuristic outperforms the existing methods. 
	For future work, the beamforming technique can be considered when the sensors and the remote estimator are equipped with multiple antennas, which can further improve the spectrum efficiency.
	
\appendix
\section{Proof of Proposition~\ref{prop:convergence}}
	{We first prove the monotonic convergence of Algorithm~\ref{algo:sca} with the following Lemma~\ref{lm:mono_converge}:
	\begin{lemma}[\cite{royden1988real}, Theorem 15]\label{lm:mono_converge}
		A monotone sequence of real numbers converges if and only if it is bounded.	
	\end{lemma}} 
	{It is easy to see that the optimal point (solution) obtained in the $(t-1)$-th iteration is also a feasible point of the optimization problem in the $t$-iteration:
	\begin{equation}
		\begin{split}
		&\left(((\eta_k^i)^{(t)}+(\gamma_k^i)^{(t)})^2-2(b_k^i)^{(t)}((\eta_k^i)^{(t)}-(\gamma_k^i)^{(t)})\right.\\ &\left.\quad+((b_k^i)^{(t)})^2\right)/{4}\\&=
		(\eta_k^i)^{(t)} (\gamma_k^i)^{(t)}\\ &\leq \left(((\eta_k^i)^{(t)}+(\gamma_k^i)^{(t)})^2-2(b_k^i)^{(t-1)}((\eta_k^i)^{(t)}-(\gamma_k^i)^{(t)})\right.\\&\left.\quad+((b_k^i)^{(t-1)})^2\right)/4\\
		&\leq \frac{(h^i_k)^{(t)}(p^i_k)^{(t)}}{\theta_i},
		\end{split}
	\end{equation}}
	That is, the optimal solution in the $(t-1)$-th iteration is also achievable in the $t$-th iteration. Therefore, the optimal solution obtained in the $t$-th iteration is no greater than that obtained in the $(t-1)$-th iteration. {Moreover, since $P_k\geq 0$, the value of the objective function is bounded.} Monotonic convergence of Algorithm~\ref{algo:sca} is hence proved.
	
	Next, we prove the converged solution satisfies all the constraints as well as the KKT conditions of Problem~\ref{relaxed_problem}. 
	
	According to~\cite{marks1978general}, the SCA-based algorithm (Algorithm~\ref{algo:sca}) can always converge to a solution satisfying the KKT conditions of Problem~\ref{relaxed_problem} when the following conditions are satisfied:
	\begin{itemize}
		\item[1)] Each non-convex constraint $f(\boldsymbol{x})\leq 0$ of the optimization problem is iteratively approximated by a convex constraint $f_{cv}(\boldsymbol{x};\hat{\boldsymbol{x}})\leq 0$, where $\hat{\boldsymbol{x}}$ is the optimal solution to the approximated problem in the previous iteration, and $f_{cv}(\boldsymbol{x};\hat{\boldsymbol{x}})$ is a convex function satisfying:
		\begin{align}
			f_{cv}(\boldsymbol{x};\hat{\boldsymbol{x}})&\geq f(\boldsymbol{x}),\\
			f_{cv}(\hat{\boldsymbol{x}};\hat{\boldsymbol{x}})&= f(\hat{\boldsymbol{x}}),\\
			\nabla f_{cv}(\boldsymbol{x};\hat{\boldsymbol{x}})|_{\boldsymbol{x}=\hat{\boldsymbol{x}}}&=\nabla f(\boldsymbol{x})|_{\boldsymbol{x}=\hat{\boldsymbol{x}}}.
		\end{align}
		
		\item[2)] The approximated convex problem satisfies the Slater's condition~\cite{boyd2004convex}.
	\end{itemize}
	
	{First, it is easy to verify condition 1) is satisfied when the approximation~(\ref{eq:approximation}) is adopted.}	
	Next, we show condition 2) is also satisfied for Problem~\ref{convex_problem}, i.e., there exists a strictly feasible solution to Problem~\ref{convex_problem}~\cite{boyd2004convex}. 
	As mentioned before, the $(t-1)$-th solution $\{\boldsymbol{\gamma}_k,\boldsymbol{\eta}_k,\boldsymbol{p}_k,P_k\}^{(t-1)}$ is also a feasible solution to the optimization problem in the $t$-th iteration. Let $\hat{p}_k^i=\phi_k^i({p}_k^i)^{(t-1)}$, where $\phi_k^i\in\mathbb{R}$ and
	\begin{align}
		\begin{split}
			&1+\frac{(\hat{\gamma}_k^i-(\gamma_k^i)^{(t-1)})(\hat{\gamma}_k^i-(\gamma_k^i)^{(t-1)}+4(\eta_k^i)^{(t-1)})}{\frac{4h_k^i p_k^i}{\theta_i}}\\&<\phi_k^i<1, \forall i,
		\end{split}\\
		&(\gamma_k^i)^{(t-1)}-4(\eta_k^i)^{(t-1)}<\hat{\gamma}_k^i<(\gamma_k^i)^{(t-1)}, \forall i.
	\end{align}
	It is easy to show that the solution $\{\hat{\boldsymbol{\gamma}}_k, \boldsymbol{\eta}_k^{(t-1)},\hat{\boldsymbol{p}}_k,P_k^{(t-1)}\}$ is a strictly feasible solution to Problem~\ref{convex_problem}. Therefore, the converged solution satisfies the KKT conditions of Problem~\ref{relaxed_problem}.


\bibliographystyle{unsrt}
\bibliography{autosam.bib}

\begin{thebibliography}{10}

\bibitem{wang2019whittle}
Jiazheng Wang, Xiaoqiang Ren, Yilin Mo, and Ling Shi.
\newblock Whittle index policy for dynamic multichannel allocation in remote
  state estimation.
\newblock {\em IEEE Transactions on Automatic Control}, 65(2):591--603, 2019.

\bibitem{wu2020optimal}
Shuang Wu, Kemi Ding, Peng Cheng, and Ling Shi.
\newblock Optimal scheduling of multiple sensors over lossy and bandwidth
  limited channels.
\newblock {\em IEEE Transactions on Control of Network Systems},
  7(3):1188--1200, 2020.

\bibitem{mo2011sensor}
Yilin Mo, Roberto Ambrosino, and Bruno Sinopoli.
\newblock Sensor selection strategies for state estimation in energy
  constrained wireless sensor networks.
\newblock {\em Automatica}, 47(7):1330--1338, 2011.

\bibitem{shi2013optimal}
Dawei Shi and Tongwen Chen.
\newblock Optimal periodic scheduling of sensor networks: A branch and bound
  approach.
\newblock {\em Systems \& Control Letters}, 62(9):732--738, 2013.

\bibitem{shi2013approximate}
Dawei Shi and Tongwen Chen.
\newblock Approximate optimal periodic scheduling of multiple sensors with
  constraints.
\newblock {\em Automatica}, 49(4):993--1000, 2013.

\bibitem{yang2015deterministic}
Chao Yang, Junfeng Wu, Xiaoqiang Ren, Wen Yang, Hongbo Shi, and Ling Shi.
\newblock Deterministic sensor selection for centralized state estimation under
  limited communication resource.
\newblock {\em IEEE transactions on signal processing}, 63(9):2336--2348, 2015.

\bibitem{asghar2017complete}
Ahmad~Bilal Asghar, Syed~Talha Jawaid, and Stephen~L Smith.
\newblock A complete greedy algorithm for infinite-horizon sensor scheduling.
\newblock {\em Automatica}, 81:335--341, 2017.

\bibitem{huang2021joint}
Lingying Huang, Junfeng Wu, Yilin Mo, and Ling Shi.
\newblock Joint sensor and actuator placement for infinite-horizon lqg control.
\newblock {\em IEEE Transactions on Automatic Control}, 67(1):398--405, 2021.

\bibitem{tse2005fundamentals}
David Tse and Pramod Viswanath.
\newblock {\em Fundamentals of wireless communication}.
\newblock Cambridge university press, 2005.

\bibitem{hossain2014evolution}
Ekram Hossain, Mehdi Rasti, Hina Tabassum, and Amr Abdelnasser.
\newblock Evolution toward 5g multi-tier cellular wireless networks: An
  interference management perspective.
\newblock {\em IEEE Wireless communications}, 21(3):118--127, 2014.

\bibitem{li2014multi}
Yuzhe Li, Daniel~E Quevedo, Vincent Lau, and Ling Shi.
\newblock Multi-sensor transmission power scheduling for remote state
  estimation under sinr model.
\newblock In {\em Proceedings of 53rd IEEE conference on decision and control},
  pages 1055--1060, 2014.

\bibitem{li2019power}
Yuzhe Li, Chung~Shue Chen, and Wing~Shing Wong.
\newblock Power control for multi-sensor remote state estimation over
  interference channel.
\newblock {\em Systems \& Control Letters}, 126:1--7, 2019.

\bibitem{ding2021interference}
Kemi Ding, Xiaoqiang Ren, Hongsheng Qi, Guodong Shi, Xiaofan Wang, and Ling
  Shi.
\newblock Interference game for intelligent sensors in cyber--physical systems.
\newblock {\em Automatica}, 129:109668, 2021.

\bibitem{dobre2011second}
Octavia~A Dobre, Ramachandran Venkatesan, Dimitrie~C Popescu, et~al.
\newblock Second-order cyclostationarity of mobile wimax and lte ofdm signals
  and application to spectrum awareness in cognitive radio systems.
\newblock {\em IEEE Journal of Selected Topics in Signal Processing},
  6(1):26--42, 2011.

\bibitem{li2007real}
Yanjun Li, Chung~Shue Chen, Ye-Qiong Song, and Zhi Wang.
\newblock Real-time qos support in wireless sensor networks: a survey.
\newblock {\em IFAC Proceedings Volumes}, 40(22):373--380, 2007.

\bibitem{zhao2017joint}
Ming-Min Zhao, Qingjiang Shi, Yunlong Cai, and Min-Jian Zhao.
\newblock Joint transceiver design for full-duplex cloud radio access networks
  with swipt.
\newblock {\em IEEE Transactions on Wireless Communications}, 16(9):5644--5658,
  2017.

\bibitem{tang2019user}
Weijun Tang and Suili Feng.
\newblock User selection and power minimization in full-duplex cloud radio
  access networks.
\newblock {\em IEEE Transactions on Signal Processing}, 67(9):2426--2438, 2019.

\bibitem{yang2021joint}
Huiwen Yang, Xinjiang Xia, Jiamin Li, Pengcheng Zhu, and Xiaohu You.
\newblock Joint transceiver design for network-assisted full-duplex systems
  with swipt.
\newblock {\em IEEE Systems Journal}, 16(1):1206--1216, 2021.

\bibitem{xia2021joint}
Xinjiang Xia, Pengcheng Zhu, Jiamin Li, Hao Wu, Dongming Wang, Yuanxue Xin, and
  Xiaohu You.
\newblock Joint user selection and transceiver design for cell-free with
  network-assisted full duplexing.
\newblock {\em IEEE Transactions on Wireless Communications},
  20(12):7856--7870, 2021.

\bibitem{coleri2002channel}
Sinem Coleri, Mustafa Ergen, Anuj Puri, and Ahmad Bahai.
\newblock Channel estimation techniques based on pilot arrangement in ofdm
  systems.
\newblock {\em IEEE Transactions on broadcasting}, 48(3):223--229, 2002.

\bibitem{8920188}
Sean Weerakkody, Omur Ozel, Yilin Mo, and Bruno Sinopoli.
\newblock {\em Resilient Control in Cyber-Physical Systems: Countering
  Uncertainty, Constraints, and Adversarial Behavior}.
\newblock Now Foundations and Trends, 2019.

\bibitem{porta2009linear}
Josep~M Porta, Lluis Ros, and Federico Thomas.
\newblock A linear relaxation technique for the position analysis of multiloop
  linkages.
\newblock {\em IEEE Transactions on Robotics}, 25(2):225--239, 2009.

\bibitem{park2015multihop}
Seok-Hwan Park, Osvaldo Simeone, Onur Sahin, and Shlomo Shamai.
\newblock Multihop backhaul compression for the uplink of cloud radio access
  networks.
\newblock {\em IEEE Transactions on Vehicular Technology}, 65(5):3185--3199,
  2015.

\bibitem{liu2017cross}
Liang Liu and Wei Yu.
\newblock Cross-layer design for downlink multihop cloud radio access networks
  with network coding.
\newblock {\em IEEE Transactions on Signal Processing}, 65(7):1728--1740, 2017.

\bibitem{liu2018transceiver}
Xiaonan Liu, Zan Li, Nan Zhao, Weixiao Meng, Guan Gui, Yunfei Chen, and
  Fumiyuki Adachi.
\newblock Transceiver design and multihop d2d for uav iot coverage in
  disasters.
\newblock {\em IEEE Internet of Things Journal}, 6(2):1803--1815, 2018.

\bibitem{marks1978general}
Barry~R Marks and Gordon~P Wright.
\newblock A general inner approximation algorithm for nonconvex mathematical
  programs.
\newblock {\em Operations research}, 26(4):681--683, 1978.

\bibitem{tsiropoulos2014radio}
Georgios~I Tsiropoulos, Octavia~A Dobre, Mohamed~Hossam Ahmed, and Kareem~E
  Baddour.
\newblock Radio resource allocation techniques for efficient spectrum access in
  cognitive radio networks.
\newblock {\em IEEE Communications Surveys \& Tutorials}, 18(1):824--847, 2014.

\bibitem{royden1988real}
Halsey~Lawrence Royden and Patrick Fitzpatrick.
\newblock {\em Real analysis}, volume~32.
\newblock Macmillan New York, 1988.

\bibitem{boyd2004convex}
Stephen Boyd, Stephen~P Boyd, and Lieven Vandenberghe.
\newblock {\em Convex optimization}.
\newblock Cambridge university press, 2004.

\end{thebibliography}

\end{document}